\documentclass{article}

\usepackage{hyperref}
\usepackage{graphicx}
\usepackage{float}
\usepackage{subfigure}
\usepackage{amsmath, amsthm}%
\newtheorem{theorem}{Theorem}
\newtheorem{lemma}{Lemma}

\newtheorem{corollary}{Corollary}

\usepackage{fullpage}%
\usepackage{amsfonts}%

\usepackage[export]{adjustbox}

\usepackage[ruled,vlined,linesnumbered, algo2e]{algorithm2e}
\usepackage{framed}

\usepackage{tikz}

\providecommand{\norm}[1]{\left\lVert#1\right\rVert}

\newcommand{\br}[1]{\left\{#1\right\}}                            

\newcommand{\REAL}{\ensuremath{\mathbb{R}}}
\newcommand{\cost}{\ensuremath{\mathrm{cost}}}
\newcommand{\eps}{\varepsilon}
\newcommand{\dist}{\mathrm{dist}}
\newcommand{\OPT}{\ensuremath{\mathrm{opt}}}

\newcommand{\q}{\{q\}}

\newcommand{\algkmean}{\textsc{$k$-Mean-Coreset}}

\newcommand{\opt}{\OPT}
\newcommand{\bigdata}{\text{Big Data}}

\title{$k$-Means for Streaming and Distributed Big Sparse Data}
\date{}
\author{
Artem Barger \and Dan Feldman \\
\small \hspace{-4.1cm}artem@barger.net  \hspace{1.2cm} dannyf.post@gmail.com\\[0.1cm]
\hspace*{-5cm} University of Haifa}

\makeatletter
\newcommand{\removelatexerror}{\let\@latex@error\@gobble}
\makeatother

\begin{document}

\maketitle

\begin{abstract}

We provide the first streaming algorithm for computing a provable approximation to the $k$-means of sparse \bigdata. Here, sparse \bigdata is a set of $n$ vectors in $\REAL^d$, where each vector has $O(1)$ non-zeroes entries, and $d\geq n$. E.g., adjacency matrix of a graph, web-links, social network, document-terms, or image-features matrices.

Our streaming algorithm stores at most $\log n\cdot k^{O(1)}$ input points in memory. If the stream is distributed among $M$ machines, the running time reduces by a factor of $M$, while communicating a total of $M\cdot k^{O(1)}$ (sparse) input points between the machines.

Our contribution is a deterministic algorithm for computing a \emph{sparse $(k,\eps)$-coreset}, which is a weighted \emph{subset} of $k^{O(1)}$ input points that approximates the sum of squared distances from the $n$ input points to \emph{every} $k$ centers, up to $(1\pm\eps)$ factor, for any given constant $\eps>0$. This is the first such coreset of size independent of both $d$ and $n$.

Existing algorithms use coresets of size at least polynomial in $d$, or project the input points on a subspace which diminishes their sparsity, thus require memory and communication $\Omega(d)=\Omega(n)$ even for $k=2$.

Experimental results on synthetic and real public dataset shows that our algorithm boost the performance of such given heuristics even in the off-line setting. Open access to our implementation is also provided~\cite{codeimpl}. 
\end{abstract}

\section{Background}
\paragraph{Clustering. }
For a given similarity measure, clustering aims at partitioning a given set of points into coherent subset. There is a lot of different clustering techniques, but most prominent and popular is Lloyd's algorithm or the $k$-means algorithm~\cite{ostrovsky2006effectiveness,arthur2007k}. The input to the classical Euclidean $k$-means problem is a set of $n$ points in $\REAL^d$, and the goal is to compute a set of $k$-centers (also points in $\REAL^d$) that minimizes the sum of squared distances over each input point to its nearest center.

More generally, there are some constraints on the location of the centers. For example, in the \emph{discrete $k$-mean} the set of centers must be subset of the input points. This version is preferable for example, when each input point is sparse (i.e., have small number of non-zeroes coordinates), and we wish that the centers will also be sparse. In other applications, such as in computer vision, the input is a set of images and we wish to compute representative $k$ images (centers). Since it is not clear how to interpret the sum of a vector that represents an image of a cat with a vector that represents an image of a dog, it is natural to ask that the set of $k$ centers will be a subset of the input images. In facility location problems~\cite{hamacher2002facility,feldman2006coresets}, the input points represent the location of clients, and the centers are called facilities. In this case we might have constraints that some of the centers will be closer to some clients.

However, most of the existing clustering algorithms, especially the ones that support streaming or large distributed data sets, assume that the dimension is significantly smaller than the number of input points. Otherwise, techniques such as PCA/SVD~\cite{lee2007nonlinear,feldman2013turning} or Johnson-Lindenstraus~\cite{johnson1984extensions} are used for reducing the dimensionality of the input. However, such projections turn sparse data into dense data, and the $k$-means of the projected points are no longer subset of the input or sparse~\cite{feldman2013turning}. In fact, few projected (dense) points in a high dimensional space might take more memory than the complete original dataset, e.g. in Wikipedia document-term $n\times d$ matrix, where $d\gg n$.

\paragraph{Sparse \bigdata. }
We consider sparse \bigdata as a set of $n$ vectors in $\REAL^d$, where each vector has $O(1)$ non-zeroes entries, and both the number of vectors $n$ and their dimensionality $d$ is asymptotically large. Hence, we cannot store all the vectors, or general linear combinations of vectors, in memory. Instead, we allow to scan these vectors in a streaming fashion, while using memory that is only logarithmic in $n$ or $d$.

In particular, unlike previous related work, this paper handle the case $d\geq n$ (e.g. $d=n^n$), where most of those papers assume $d\ll n$. Our algorithm also supports input such as in NoSQL or Matlab sparse matrices, where the row $i$ is given as a stream of pairs $(j_1,val_1), (j_2,val_2),\cdots$ where $(j_m,val_m)$ means that the $j_m$th entry of the row is $val_m$, and the rest are zeroes by default.

Sometimes a cloud or a large set of machines are available for handling \bigdata. In this case, the data is distributed among them and we wish not only to process the data in parallel, but also to minimize communication between the machines. Finally, GPUs (Graphical Processing Units) are used to boost the performance of the algorithm, but only if the algorithm uses a limited set of simple operations that are supported by the GPU.

All these models are supported by our algorithm, via the framework of sparse coresets that is explained below.

\paragraph{Coresets} Given a set of $n$ points in $\REAL^d$, and an error parameter $\eps>0$, a coreset in this paper is a small set of weighted points in $\REAL^d$, such that the sum of squared distances from the original set of points to any set of $k$ centers in $\REAL^d$ can be approximated by the sum of weighted squared distances from the points in the coreset.
The output of running an existing clustering algorithm on the coreset would then yield approximation to the output of running the same algorithm on the original data, by the definition of the coreset.

Coresets were first suggested in~\cite{AgaHarVar04} as a way to improve the theoretical running time of existing algorithms.
Moreover, a coreset is a natural tool for handling \bigdata using all the computation models that are mentioned in the previous section. This is mainly due to the merge-and-reduce tree approach that was suggested in~\cite{sariela, bentley1980decomposable} and is formalized in~\cite{FSS13}: coresets can be computed independently for subsets of input points, e.g. on different computers, and then be merged and re-compressed again. Such a binary compression tree can also be computed using one pass over a possibly unbounded streaming set of points, where in every moment only $O(\log n)$ coresets exist in memory for the $n$ points streamed so far. Here the coreset is computed only on small chunks of points, so a possibly inefficient coreset construction still yields efficient coreset constructions for large sets; see Fig.~\ref{fig:tree}. Note that the coreset guarantees are preserved while using this technique, while no assumptions are made on the order of the streaming input points.

In practice, this technique can be implemented easily using the map-reduce approach of modern software for \bigdata such as Hadoop~\cite{hadoop2009hadoop}.

The fact that the coreset approximates every set of $k$ centers, allows us to use the above merge-and-reduce technique where the optimal solution keeps changing, and also to solve the $k$-means problem under different constraints or a small set of candidate centers.

\section{Related Work}\label{related}
We summarize existing coresets constructions for $k$-mean queries, as will be formally defined in Section~\ref{sec:coreset}.

\paragraph{Importance Sampling.} Following a decade of research, coreset of size polynomial in $d$ were suggested in~\cite{KeChen06}. An improved version of size $O(dk^2/\eps^2)$ was suggested in~\cite{LS10} which is a special case of the algorithms in~\cite{feldman2007ptas}. The construction is based on computing an approximation to the $k$-mean of the input points (with no constraints on the centers) and then sample points proportionally to their distance to these centers. Each chosen point is then assigned a weight that is inverse proportional to this distance.  The probability of failure in this algorithms reduces exponentially with the input size. Coresets of size $O(dk/\eps^2)$, i.e., linear in $k$, were suggested in~\cite{FL11}, however the weight of a point may be negative or a function of the given query.
For the special case $k=1$, such coresets of size $O(1/\eps^2)$ were suggested in~\cite{InaKatIma94} using uniform sampling.

\textbf{Projection based coresets.} Coresets of size $O(k/\eps)$ that are based on projections on low-dimensional subspaces that diminishes the sparsity of the input data were recently suggested in~\cite{mcohen} by improving the analysis in~\cite{FSS13}. Other type of weak coresets approximates only the optimal solution, and not every $k$ centers. Such randomized coresets of size independent of $d$ and only polynomial in $k$ were suggested in~\cite{feldman2007ptas} and simplified in~\cite{FL11}.

\textbf{Deterministic Constructions.} The first coresets for $k$-means~\cite{sariela, sarielb} were based on partitioning the data into cells, and take a representative point from each cell to the coreset, as is done in hashing or Hough transform~\cite{ballard1981generalizing}. However, these coresets are of size at least $k/\eps^{O(d)}$, i.e., exponential in $d$. Deterministic coreset constructions of size $k^{O(k/\eps)}$, i.e.,  independent of $d$ but exponential in $k$, were suggested in~\cite{FSS13} by recursively computing $k$-means clustering of the input points.  Uniform sampling is then used for replacing each mean of a cluster by a subset of its points. By Markov's inequality, the probability of failure for these algorithms reduces only linearly with the input size. Therefore they cannot be used in the streaming or distributed setting, where we need to compute union of $O(n)$ core-sets.
Our technique improves this result by suggesting a fully deterministic coreset construction of size $k^{O(1/\eps^2)}$, i.e., polynomial in $k$.


\section{Our contribution\label{our}}
Our main technical result is a deterministic algorithm that computes a $(k,\eps)$-coreset of size $k^{O(1)}$ for a set of points $P$ and every constant error parameter $\eps>0$; see Theorem~\ref{thm1} and Corollary~\ref{cor} for details and exact bounds. This is the first coreset which uses number of memory words that is independent of both $n$ and $d$, and only polynomial in $k$ where each point of $P$ has $O(1)$ non-zeroes entries.
Using this coreset with the merge-and-reduce technique above, we achieve the following deterministic algorithms:

\begin{enumerate}
    \item An algorithm that computes a $(1+\eps)$ approximation for the $k$-means of a set $P$ that is distributed (partitioned) among $M$ machines, where each machine needs to send only $k^{O(1)}$ input points to the main server at the end of its computation.
    
\item A streaming algorithm that, after one pass over the data and using $k^{O(1)}\log n$ space returns an $O(\log n)$-approximation to the $k$-means of $P$.
\item Description of how to use our algorithm to boost both the running time and quality of any existing $k$-means heuristic using only the heuristic itself, even in the classic off-line setting.
\item Experimental results on real world data-sets and open-code that demonstrate how we applied our coreset to boost the performance of the popular Llooyd's $k$-means heuristic~\cite{ostrovsky2006effectiveness, arthur2007k}.
\end{enumerate}

\paragraph{Running time. } The overall running time of our algorithm is $n\log (n)\cdot k^{O(k)}$ for the set $P$ of $n$ points seen so far in a possibly unbounded stream, where each point has $O(1)$ non-zeroes entries. Computing a $(1+\eps)$-approximation to the $k$-means on the coreset takes time $|S|^{O(k)}$ where $|S|$ is the coreset size.
Running time that is exponential in $k$ is unavoidable for any $(1+\eps)$-approximation algorithm that solves $k$-means, even in the planar case $(d=2)$~\cite{mahajan2009planar}. Our main contributions thus a coreset construction that uses memory that is independent of $d$ and running time that is near-linear in $n$. This is a new result even for the case $k=2$.  To handle large values of $k$ in our experiments, our algorithm calls popular heuristics instead of computing the optimal $k$-means during the coreset construction and on the resulting coreset itself.

\section{Notation and Main Result\label{sec:coreset}}
The input to our problem is a set $P'$ of $n$ points in $\REAL^d$, where each point $p\in P'$ includes a multiplicative weight $u(p)>0$. In addition, there is an additive weight $\rho>0$ for the set. Formally, a \emph{weighted set} in $\REAL^d$ is a tuple $P=(P', u, \rho)$, where $P' \subseteq \REAL^d$, $u : P' \rightarrow [0,\infty)$. In particular, an \emph{unweighted} set has a unit weight $u(p)=1$ for each point, and a zero additive weight.

\paragraph{$k$-Mean clustering.} For a given set $Q=\br{q_1,\cdots,q_k}$ of $k\geq1$ centers (points) in $\REAL^d$, the Euclidean distance from a point $p\in \REAL^d$ to its closest center in $Q$ is denoted by
$\dist(p,Q)=\min_{q\in Q}\norm{p-q}_2$. The sum of these weighted squared distances over the points of $P$ is denoted by
        \[
            \cost(P ,Q) := \sum_{p \in P'}u(p)\cdot\dist^2(p,Q) + \rho,
        \]
If $P$ is an unweighted set, this cost is just the sum of squared distances over each point in $P'$ to its closest center in $Q$.

Let $P'_{i}$ denote the subset of points in $P'$ whose closest center in $Q$ is $q_i$, for every $i\in[m]=\br{1,\cdots,m}$. Ties are broken arbitrarily. This yields a partition $\br{P'_1,\cdots,P'_k}$ of $P'$ by $Q$.
More generally, the \emph{partition of $P$ by $Q$} is the set $\br{P_1,\cdots,P_k}$ where $P_i=(P'_i,u_i,\rho/k)$, and $u_i(p)=u(p)$ for every $p\in P'_i$ and every $i\in[k]$.

A set $Q_k$ that minimizes this weighted sum $\cost(P,Q)$ over every set $Q$ of $k$ centers in $\REAL^d$ is called the \emph{$k$-mean} of $P$. The $1$-means $\mu(P)$ of $P$ is called the centroid, or the center of mass, since
\[
\mu(P)=\frac{1}{\sum_{p'\in P'} u(p')}\sum_{p\in P'}u(p)\cdot p.
\]
We denote the cost of the $k$-mean of $P$ by $\OPT(P,k):=\cost(P,Q_k)$.

\paragraph{Coreset.}
Computing the $k$-mean of a weighted set $P$ is the main motivation of this paper.
To this end, we wish to compute another weighed set $S=(S', w, \phi)$ where $S'$ is small set (core-set) $\REAL^d$ that can be used to approximate $\cost(P,Q)$ for any set $Q$ of $k$ points. In particular, a set $Q$ that minimizes the weighted cost to $S$ must also approximately minimize the cost to $P$, over such sets $Q$ in $\REAL^d$.

Formally,  let $\eps>0$ be an error parameter. The weighted set $S=(S',w, \phi)$ is a $(k,\eps)$-{coreset} for $P$, if for every set $Q \subset \REAL^d$ of $|Q|=k$ centers we have
\[
(1-\eps)\cost(P,Q) \leq \cost(S,Q) \leq (1+\eps)\cost(P,Q),
\]
That is,
\[
    \begin{split}
    (1 - \eps)\sum_{p \in P} u(p)\dist^2(p, Q) +\rho \leq \sum_{p \in S}w(p)\cdot \dist^2(p, Q) + \phi
    \leq (1 + \eps) \sum_{p \in P}u(p)\cdot\dist^2(p, Q) + \rho.
    \end{split}
\]

To handle streaming data we will need to compute ``coresets for union of coresets", which is the reason that we assume that both the input $P$ and its coreset $S$ are weighted sets.

\paragraph{Sparse coresets.} Unlike previous work, we add the constraints that if each point in $P$ is sparse, i.e., has few non-zeroes coordinates, then the set $S$ will also be sparse. Formally, the \emph{maximum sparsity} $s(P):=\max_{p\in P}\norm{p}_0$ of $P$ is the maximum number of non zeroes entries $\norm{p}_0=|\br{i\mid p_i\neq 0, i\in[d]}|$ over every point $p$ in $P$.

In particular, if each point in $S$ is a linear combination of at most $\alpha$ points in $P$, then $s(S)\leq \alpha \cdot s(P)$. In addition, we would like that the set $S$ will be of size independent of both $n=|P|$ and $d$.

We can now state the main result of this paper.
\begin{theorem}[Small sparse coreset]\label{thm1}
For every weighted set $P=(P',u,\rho)$ in $\REAL^d$, $\eps>0$ and an integer $k\geq1$, there is a $(k,\eps)$-coreset $S=(S',w,\phi)$ of size $|S|=k^{O(1/\eps^2)}$ where each point in $S'$ is a linear combination of $O(1/\eps^2)$ points from $P'$. In particular, the maximum sparsity of $S'$ is $s(P)/\eps^2$.
\end{theorem}

By plugging this result to the traditional merge-and-reduce tree below, it is straight-forward to compute a coreset using one pass over a stream of points. Previous results computed such coresets using $O(k/\eps)$ points but $O(dk/\eps)$ memory words (coordinates of points) which is inefficient when, say, $d\geq n$. In contrast, the coreset below is computed in memory of size that is independent of $d$. Similarly, when the input is distributed among few machines, previous results had to communicate coresets of size at least linear in $d$ between the machines, while communicating the coreset below will take number of bits that is independent of $d$.

\begin{corollary}\label{cor}
A $(k,\eps \log n)$ coreset $(S,w,\phi)$ of size $|S|=\log(n)\cdot k^{O(1/\eps^2)}$ and maximum sparsity $s(P)/\eps^2$
 can be computed for the set $P$ of the $n$ points seen so far in an unbounded stream, using $|S|\cdot s(P)/\eps^2$
 memory words. The insertion time per point in the stream is $\log (n)\cdot 2^{(k/\eps)^{O(1)}}$. If the stream is distributed uniformly to $M$ machines, then the amortized insertion time per point is reduced by a (multiplicative) factor of $M$ to $(1/M)\log (n)\cdot 2^{(k/\eps)^{O(1)}}$. The coreset for the union of streams can then be computed by communicating the $M$ coresets to a main server.
\end{corollary}

\begin{figure}[htbp]
\centering
{{\includegraphics[width=0.5\textwidth, height=4cm]{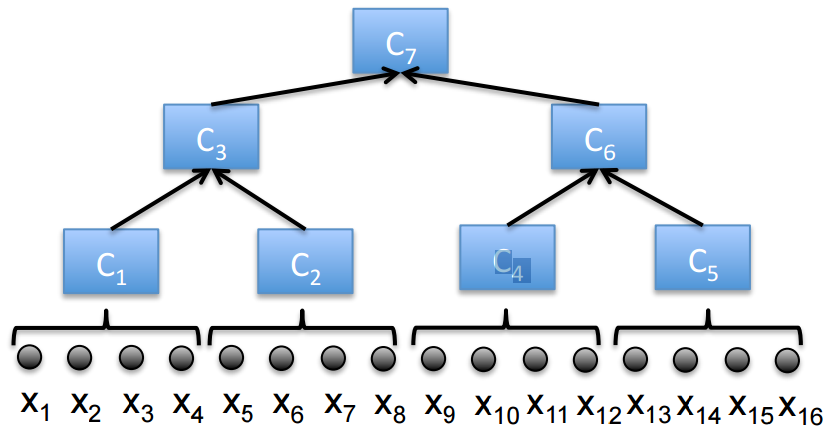}}
 }
 \caption{\small \it
 Tree construction for generating coresets in parallel or from data streams~\cite{sariela}. Black arrows indicate ``merge-and-reduce'' operations. The intermediate coresets $C_{1},\dots,C_{7}$ are numbered in the order in which they would be generated in the streaming case. In the parallel case, $C_{1},C_{2},C_{4}$ and $C_{5}$ would be constructed in parallel, followed by  $C_{3}$ and $C_{6}$, finally resulting in $C_{7}$. The Figure is from~\cite{feldman2011scalable}.
 \label{fig:tree}
 }
\end{figure}

\section{Coreset Construction}
 Our main coreset construction algorithm $\algkmean(P,k,\eps)$ gets a set $P$ as input, and returns a $(k,\eps)$-coreset $(S,w)$; see Algorithm~\ref{algk}.
 
To obtain running time that is linear in the input size, without loss of generality, we assume that $P$ has $|P|=k^{O(1/\eps^2)}$ points, and that the cardinality of the output $S$ is $|S|\leq |P|/2$. This is thanks to the traditional merge-and-reduce approach: given a stream of $n$ points, we apply the coreset construction only on subsets of size $2\cdot |S|$ from $P$ during the streaming and reduce them by half. See Fig.~\ref{fig:tree} and e.g.~\cite{feldman2011scalable,FSS13} for details.

\paragraph{Algorithm overview.}
In Line~\ref{one} we compute the smallest integer $m=k^t$ such that the cost $\OPT(P,m)$ of the $m$-means of $P$ is close to the cost $\OPT(P,mk)$ of the $(mk)$-means of $P$.
In Line~\ref{two} we compute the corresponding partition $\br{P_1,\cdots,P_{m}}$ of $P$ by its $m$-means $Q_m=\br{q_1,\cdots,q_m}$.
In Line~\ref{four} a $(1,\eps)$-sparse coreset $S_i$ of size $O(1/\eps^2)$ is computed for every $P_i$, $i\in[m]$.
This can be done deterministically e.g. by taking the mean of $P_i$ as explained in Lemma~\ref{onelemma} or by using a gradient descent algorithm, namely Frank-Wolfe, as explained in~\cite{onecenter} which also preserves the sparsity of the coreset as desired by our main theorem. The output of the algorithm is the union of the means of all these coresets, with appropriate weight, which is the size of the coreset.

\paragraph{Comments.}
Line~\ref{one} is the only line in the algorithm that takes time exponential in $k$. As stated in Section~\ref{our} this is unavoidable for a $(1+\eps)$-approximation, where $\eps$ is an arbitrarily small constant. Still, the memory that is required by our algorithm is polynomial in $k$, and the running time is near-linear in $n$. In Section~\ref{sec:xp} this line will be replaced by $k$-means heuristic or $\Omega(1)$-approximation to yield a practical algorithm for large values of $k$.

\begin{algorithm2e}[!h]
    \caption{$\algkmean(P,k,\eps)$\label{algk}}
\begin{tabbing}
\textbf{Input:} \=A weighted set $P$ of points in $\REAL^d$, integer $k\geq 1$, and error parameter $\eps\in(0,1/4)$.\\
\textbf{Output:} \=A $(k,\eps)$-coreset $S$ for $P$ that satisfies Theorem~\ref{thm1}.
\end{tabbing}
\vspace{-0.3cm}
\nl    \label{one}Compute the smallest integer $t\geq0$ such that $\OPT(P,k^t)-\OPT(P,k^{t+1})\leq \eps^2 \cdot \OPT(P,k)$\\
\tcc{$\OPT(P,m)$ is the $m$-mean cost of $P$.}
 Set $m\gets k^t$\\
   Set $\{P_1, \dots , P_{m}\}\gets$ the partition of $P$ by $\OPT(P,m)$\label{two}\\
    \label{foreach}\For{$i:=1$ \textbf{\emph{to}} $m$} {
		Compute a $(1,\eps)$-coreset $S_i=(S'_i,w_i,\phi_i)$ for $P_i$\label{four}\\
        Set $w(\mu(S_i)) \leftarrow \sum_{p\in S'_i} w_i (p)$
     }
     Set $S' \leftarrow \bigcup_{i=1}^m \mu(S_i)$ \\
     Set $\phi \leftarrow \sum_{i=1}^m \cost(S_i,\mu(S_i))$ \\
     Set $S\gets (S', w, \phi)$ \\
	\Return $S$
\end{algorithm2e}

\section{Proof of Correctness}
In this section we prove that a call to $\algkmean(P,k,\eps)$ indeed returns a small $(k,\eps)$-coreset; see Algorithm~\ref{algk}.
We use the variable names as defined in the algorithm (e.g. $t$, $m$, $S$) and consider their corresponding values during the last line of the algorithm.
We also identify the input weighted set $P$ by $P=(P',u,\rho)$.

The first lemma states the common fact that the sum of squared distances of a set of point to a center is the sum of their squared distances to their center of mass, plus the squared distance to the center (the variance of the set).
\begin{lemma}\label{onelemma}
For every $x\in\REAL^d$
\[
\cost(P,x)=\cost(P,\mu(P))+\norm{\mu(P)-x}^2\sum_{p\in P}u(p).
\]
\end{lemma}
\begin{proof}
We have
\[
\begin{split}
\cost(P,x)-\rho
&=\sum_{p\in P}u(p)\norm{p-x}^2
=\sum_{p\in P}u(p)\norm{(p-\mu(P))+(\mu(P_i)-x)}^2\\
&=\sum_{p\in P}u(p)\norm{p-\mu(P)}^2 +\sum_{p\in P}u(p)\norm{\mu(P)-x}^2+ 2(\mu(p)-x)\cdot \sum_{p\in P}u(p)(p-\mu(P)).
\end{split}
\]
The last term equals zero since $\mu(P)=\cfrac{1}{\sum_{p'\in P}u(p')}\cdot \sum_{p\in P}u(p)\cdot p$, and thus
\[
\begin{split}
\sum_{p\in P}u(p)(p-\mu(P))
&=\sum_{p\in P}u(p)\cdot p-\sum_{p\in P}u(p)\mu(P)
=\sum_{p\in P}u(p)\cdot p-\sum_{p\in P}u(p)\cdot p=0.
\end{split}
\]
Hence,
\[
\cost(P,x)
=\rho+\sum_{p\in P}u(p)\norm{p-\mu(P)}^2 +\sum_{p\in P}u(p)\norm{\mu(P)-x}^2
=\cost(P,\mu(P))+\norm{\mu(P)-x}^2\sum_{p\in P}u(p).
\]
\end{proof}

The second lemma shows that assigning all the points of $P$ to the closest center in $Q$ yields a small multiplicative error if the $1$-mean and the $k$-mean of $P$ has roughly the same cost. If $t=0$, this means that we can approximate $\cost(P,Q)$  using only one center in the query; see Line~\ref{one} of Algorithm~\ref{algk}. Note that $(1+2\eps)/(1-2\eps)\leq 1+4\eps$ for $\eps<1/4$.
\begin{lemma}\label{main}
For every set $Q\subseteq \REAL^d$ of $|Q|=k$ centers we have
\begin{equation}\label{eq55}
 \cost(P,Q)\leq \min_{q \in Q}\cost(P, \q) \leq \cost(P,Q)\cdot\frac{1+2\eps}{1-2\eps}+\frac{\opt(P,1)-\opt(P,k)}{(1-2\eps)\eps}.
\end{equation}
\end{lemma}
\begin{proof}
Let $q^*$ denote a center that minimizes $\cost(P, \q)$ over $q\in Q$.
The left inequality of~\eqref{eq55} is then straight-forward since
\begin{equation}\label{left}
\begin{split}
\cost(P,Q)-\min_{q \in Q}\cost(P, \q)
&=\sum_{p\in P'}\min_{q\in Q}u(p)\norm{p-q}^2-\sum_{p\in P'}u(p)\norm{p-q^*}^2\\
&=\sum_{p\in P'}u(p)\left(\min_{q\in Q}\norm{p-q}^2-\norm{p-q^*}^2\right)
\leq 0.
\end{split}
\end{equation}
It is left to prove the right inequality of~\eqref{eq55}. Indeed, for every $p\in P'$, let $q_p\in Q$ denote the closest point to $p$ in $Q$. Ties are broken arbitrarily. Hence,
\[
\min_{q \in Q}\cost(P, \q)-\cost(P,Q)
=\sum_{p\in P'}u(p)\norm{p-q^*}^2-\sum_{p\in P'}u(p)\norm{p-q_p}^2.
\]
Let $\br{P_1,\cdots,P_k}$ denote the partition of $P$ by $Q=\br{q^1,\cdots,q^k}$, where $P_i$ are the closest points to $q^i$ for every $i\in[k]$; see Section~\ref{sec:coreset}.
For every $p\in P'_i$, let $q_p^*=\mu(P_i)$. Hence,
\begin{align}
&\sum_{p\in P'_i}u(p)\norm{p-q^*}^2-\sum_{p\in P'_i}u(p)\norm{p-q_p}^2
\label{aa}=\sum_{p\in P'_i}u(p)\norm{p-\mu(P_i)}^2+\norm{\mu(P_i)-q^*}^2\sum_{p\in P'_i}u(p)\\
&\label{cc}-\left(\sum_{p\in P'_i}u(p)\norm{p-\mu(P_i)}^2+\norm{\mu(P_i)-q^i}^2\sum_{p\in P'_i}u(p)\right)\\
&\label{bb}=\sum_{p\in P'_i}u(p)\left(\norm{q_p^*-q^*}^2
-\norm{q_p^*-q_p}^2\right),
\end{align}
where in~\eqref{aa} and~\eqref{cc} we substituted $x=\mu(P_i)$ and $x=q_i$ respectively in Lemma~\ref{onelemma}, and in~\eqref{bb} we use the fact that $q_p^*=\mu(P_i)$ and $q_p=q^i$ for every $p\in P_i$.
Summing~\eqref{bb} over $i\in[k]$ yields
\begin{align}
&\nonumber\sum_{p\in P'}u(p)\norm{p-q^*}^2-\sum_{p\in P'}u(p)\norm{p-q_p}^2
=\sum_{p\in P'}u(p)\left(\norm{q_p^*-q^*}^2
-\norm{q_p^*-q_p}^2\right)\\
&=\sum_{p\in P'}u(p)\left(
\norm{(q_p^*-\mu(P))+(\mu(P)-q^*)}^2
-\norm{(q_p^*-\mu(P))+(\mu(P)-q_p)}^2\right)\\
&\label{kkc}=\sum_{p\in P'}u(p)\left(\norm{\mu(P)-q^*}^2-\norm{\mu(P)-q_p}^2\right)\\
&-2\sum_{p\in P'}u(p)(q_p^*-\mu(P))(q^*-q_p).
\end{align}

To bound~\eqref{kkc}, we substitute $x=q^*$ and then $x=q$ in Lemma~\ref{onelemma}, and obtain that for every $q\in Q$
\[
\begin{split}
\left(\norm{\mu(P)-q^*}^2-\norm{\mu(P)-q}^2\right)\sum_{p\in P'}u(p)
&=\cost(P, \{q^*\})-\cost(P,\mu(P))
-\left(\cost(P, \q)-\cost(P,\mu(P))\right)\\
&=\cost(P, \{q^*\})-\cost(P, \q)\leq 0.
\end{split}
\]
where the last inequality is by the definition of $q^*$. This implies that for every $p\in P'$,
\[
\norm{\mu(P)-q^*}^2-\norm{\mu(P)-q_p}^2\leq 0.
\]
Plugging the last inequality in~\eqref{kkc} yields
\begin{align}
&\sum_{p\in P'}u(p)\norm{p-q^*}^2-\sum_{p\in P'}u(p)\norm{p-q_p}^2\leq -2\sum_{p\in P'}u(p)(q_p^*-\mu(P))(q^*-q_p)\\
\label{cau}&\leq 2\sum_{p\in P'}u(p)\norm{q_p^*-\mu(P)}\norm{q^*-q_p}
=\sum_{p\in P'}u(p)\cdot 2\cdot\frac{\norm{q_p^*-\mu(P)}}{\sqrt{\eps}}\cdot \sqrt{\eps} \norm{q^*-q_p}\\
\label{ggt}&\leq \sum_{p\in P'}u(p) \left(\frac{\norm{q_p^*-\mu(P)}^2}{\eps}
+\eps\norm{q^*-q_p}^2\right)\\
&\label{kka}=\frac{1}{\eps}\sum_{p\in P'}u(p)\norm{q_p^*-\mu(P)}^2+\eps\sum_{p\in P'}u(p)\norm{q^*-q_p}^2,
\end{align}
where~\eqref{cau} is by Cauchy-Schwartz inequality, and in~\eqref{ggt} we use the fact that $2ab\leq a^2+b^2$ for every $a,b\geq 0$.

To bound the left term of~\eqref{kka} we use the fact $q_p^*=\mu(P_i)$ and substitute $x=\mu(P)$, $P=P_i$ in Lemma~\ref{onelemma} for every $i\in[k]$ as follows.
\begin{equation}\label{aacd}
\begin{split}
&\sum_{p\in P'}u(p)\norm{q_p^*-\mu(P)}^2
=\sum_{i=1}^k \norm{\mu(P_i)-\mu(P)}^2\sum_{p\in P'_i}u(p)
=\sum_{i=1}^k \left(\sum_{p\in P'_i}u(p)\norm{p-\mu(P)}^2
- \sum_{p\in P'_i}u(p)\norm{p-\mu(P_i)}^2\right)\\
&=\sum_{p\in P'}u(p)\norm{p-\mu(P)}^2
- \sum_{i=1}^k \sum_{p\in P'_i}u(p)\norm{p-\mu(P'_i)}^2
\leq \opt(P,1)-\opt(P,k).
\end{split}
\end{equation}

To bound the right term of~\eqref{kka} we use $(a-b)^2\leq a^2+b^2+2|ab|\leq 2a^2+2b^2$ to obtain
\[
\begin{split}
\sum_{p\in P'}u(p)\norm{q^*-q_p}^2
&=\sum_{p\in P'}u(p)\norm{(q^*-p)+(p-q_p)}^2\\
&\leq \sum_{p\in P'}u(p)\left(2\norm{q^*-p}^2+2\norm{p-q_p}^2\right)
=2\cdot\left(\cost(P,\{q^*\})+\cost(P,Q))\right).
\end{split}
\]

Plugging~\eqref{aacd} and the last inequality in~\eqref{kka} yields
\[
\begin{split}
\cost(P,\{q^*\})-\cost(P,Q)
&=\sum_{p\in P'}u(p)\norm{p-q^*}^2-\sum_{p\in P'}u(p)\norm{p-q_p}^2 \\
&\leq \frac{\opt(P,1)-\opt(P,k)}{\eps}+2\cdot\eps\left(\cost(P,\{q^*\})+\cost(P,Q)\right).
\end{split}
\]
Rearranging,
\[
\begin{split}
\cost(P,\{q^*\})&\leq \cost(P,Q)\cdot\frac{1+2\eps}{1-2\eps}+\frac{\opt(P,1)-\opt(P,k)}{(1-2\eps)\eps}\end{split}
\]
Together with~\eqref{left} this proves Lemma~\ref{main}.
\end{proof}

\begin{lemma}\label{1coreset}
Let $S$ be a $(1, \eps)$-coreset for a weighted set $P$ in $\REAL^d$.
     Let $Q \subseteq \REAL^d$ be a finite set. Then
    \begin{equation}\label{eqlemmaone}
        (1-\eps) \min_{q \in Q}\cost(P, \q)  \leq \min_{q \in Q}\cost (S, \q)  \leq (1+\eps) \min_{q \in Q} \cost(P, \q)
    \end{equation}
    \label{coresetMin}
\end{lemma}

\begin{proof}
    Let $q_P \in Q$ be a center such that $\cost(P,\{q_P\})=\min_{q \in Q}\cost(P, \q)$, and let
$q_S \in Q$ be a center such that $\cost(S, \{q_S\})=\min_{q \in Q}\cost(S, \q)$. The right side of~\eqref{eqlemmaone} is bounded by
\[
        \min_{q \in Q}\cost(S, \q)
        = \cost(S, \{q_S\})
        \leq \cost (S, \{q_P\})
         \leq (1 + \eps)\cost(P, \{q_P\})
        = (1 + \eps)\min_{q \in Q}\cost(P, \q), \label{finCoresetA}
\]
  where the first inequality is by the optimality of $q_S$, and the second inequality is since $S$ is a coreset for $P$.
  Similarly, the left hand side of~\eqref{eqlemmaone} is bounded by
\[
\begin{split}
(1 - \eps)\min_{q \in Q}\cost(P, \q)
=(1 - \eps)\cost(P, \{q_P\})
& \leq (1 - \eps)\cost(P, \{q_S\})
\leq (1 - \eps)(1+\eps)\cost(S, \{q_S\})\\
&=(1-\eps^2) \min_{q \in Q}\cost(S, \q)
\leq \min_{q \in Q}\cost(S, \q).
\end{split}
\]
where the last inequality follows from the assumption $\eps<1$.
\end{proof}

\begin{lemma}\label{mainthm}
Let $S$ be the output of a call to $\algkmean(P, k,\eps)$. Then $S$ is a $(k,15\eps)$-coreset for $P$.
\end{lemma}
\begin{proof}
By replacing $P$ with $P_i$ in Lemma~\ref{onelemma} for each $i\in [m]$ it follows that
        \begin{equation*}
            \cost(P_i,Q)\leq \min_{q \in Q}\cost(P_i, Q)
                \leq \cost(P_i,Q)\cdot \frac{1+2\eps}{1-2\eps}+\frac{\opt(P_i,1)-\opt(P_i,k)}{(1-2\eps)\eps}.
        \end{equation*}
Summing the last inequality over each $P_i$ yields
\begin{equation}\label{ee}
\begin{split}
            \cost(P,Q)\leq \sum_{i=1}^m \min_{q \in Q}\cost(P_i, Q)                
   &\leq             \cost(P,Q)\cdot  \frac{1+2\eps}{1-2\eps}+\frac{1}{(1-2\eps)\eps}
   \sum_{i=1}^m \left(\opt(P_i,1)-\opt(P_i,k)\right).
    \end{split}
\end{equation}
Since $\br{P_1,\cdots,P_m}$ is the partition of the $m$-means of $P$ we have $\sum_{i=1}^m \opt(P_i,1)=\opt(P,m)$.
By letting $Q_i$ be the $m$-means of $P_i$ we have \[
\sum_{i=1}^m \opt(P_i,k)=\sum_{i=1}^m \cost(P_i,Q_i)\geq  \sum_{i=1}^m \cost(P_i,\cup_{j=1}^m Q_j)=\cost(P,\cup_{j=1}^m Q_j)\geq \opt(P,mk).
\]
Hence,
\[
\sum_{i=1}^m \left(\opt(P_i,1)-\opt(P_i,k)\right)
\leq \opt(P,m)-\opt(P,mk)\leq \eps^2\opt(P,k)\leq \eps^2\cost(P,Q),
\]
where the second inequality is by Line~\ref{one} of the algorithm.
Plugging the last inequality in~\eqref{ee} yields
\begin{equation}\label{eq3}
\begin{split}
            \cost(P,Q)\leq \sum_{i=1}^m \min_{q \in Q}\cost(P_i, Q)
   &\leq \cost(P,Q)\cdot  \frac{1+3\eps}{1-2\eps}.
    \end{split}
\end{equation}
Using Lemma~\ref{coresetMin}, for every $i\in[m]$
        \begin{equation*}
            (1-\eps) \min_{q \in Q}\cost(P_i, \q) \leq \min_{q \in Q}\cost(S_i, \q) \leq (1+\eps) \min_{q \in Q} \cost(P_i, \q)
        \end{equation*}

        By summing over $i \in [m]$ we obtain

\[
            (1-\eps) \sum_{i=1}^{m}\min_{q \in Q}\cost(P_i, \q) \leq \sum_{i=1}^{m} \min_{q \in Q}\cost(S_i, \q) \leq (1+\eps) \sum_{i=1}^{m} \min_{q \in Q}\cost(P_i, \q).  \label{sumS_i}
        \]
By this and Lemma~\ref{onelemma}
\[
            (1-\eps) \sum_{i=1}^{m}\min_{q \in Q}\cost(P_i, \q) \leq \cost(S,Q) 
            \leq (1+\eps) \sum_{i=1}^{m} \min_{q \in Q}\cost(P_i, \q).  \label{sumS_i}
\]
Plugging the last inequality in~\eqref{eq3} yields
\begin{equation}
\begin{split}
            (1-\eps)\cost(P,Q)&\leq (1-\eps)\sum_{i=1}^m \min_{q \in Q}\cost(P_i, Q)\\
            &\leq \cost(S,Q)\\
            &\leq (1+\eps) \sum_{i=1}^{m} \min_{q \in Q}\cost(P_i, \q)
            \leq (1+\eps)\cost(P,Q)\cdot \frac{1+3\eps}{1-2\eps}
            \leq (1+15\eps)\cost(P,Q).
    \end{split} 
\end{equation}
Hence, $S$ is a $(k,15\eps)$ coreset for $P$.
\end{proof}

\begin{lemma}\label{l2}
There is an integer $t< 1+1/\eps^2$ such that
\begin{equation}\label{toprove}
\OPT(P,k^t)-\OPT(P,k^{t+1})\leq \eps^2 \cdot \OPT(P,k).
\end{equation}
\end{lemma}

    \begin{proof}
        Contradictively assume that~\eqref{toprove} does not hold for every integer $i< 1+ 1/\eps^2$.
        Hence, 
        \[
            \begin{split}
                 \OPT(P,k)-\OPT(P,k^{\lceil 1/\eps^2\rceil+1})
                    =\sum_{i=1}^{\lceil 1/\eps^2\rceil}\big(\OPT(P,k^i) - \OPT(P,k^{i+1})\big)
                > \lceil 1/\eps^2\rceil\cdot \eps^2\OPT(P,k)\geq \OPT(P,k).
            \end{split}
        \]
        Contradiction, since $\OPT(P,k^{\lceil 1/\eps^2\rceil+1})\geq 0$.
    \end{proof}

Using the mean of $P_i$ in Line~\ref{four} of the algorithm yields a $(1,\eps)$-coreset $S_i$ as shown in Lemma~\ref{onelemma}.
The resulting coreset is not sparse, but gives the following result.
\begin{theorem}
There is $m\leq k^{1/\eps^2}$ such that the $m$-means of $P$ is a $(k,15\eps)$-coreset for $P$.
\end{theorem}

\begin{proof}[\textbf{Proof of Theorem~\ref{thm1}}:]
We compute $S_i$ a $(1,\eps)$ mean coreset for $1$-mean of $P_i$ at line~\ref{four} of Algorithm~\ref{algk} by using Frank-Wolfe~\cite{onecenter} algorithm. It follows that $|S_i|=O(1/\eps^2)$ for each $i$, therefore the overall sparsity of $S$ is $s(P)/\eps^2$. This and Lemma~\ref{mainthm} concludes the proof.
\end{proof}

\section{Why does it work?}\label{secphen}
In this section we try to give an intuition of why our coreset construction yields a smaller error for the same number of samples, compared to existing coreset constructions. Roughly, this is mainly due to the ``cost of independent sampling" that is used by the existing smallest coreset constructions, namely the sensitivity/importance sampling approach~\cite{KeChen06, LS10,FL11, feldman2007ptas}.

In Fig.~\ref{fig:Ex1} the input is a set of 16 points on the plane that is distributed over $7$ clusters that are relatively far from each other. Each cluster consists of a single point, except for one cluster that has $16-6=10$ points.
Given a ``budget" (coreset size) of $m\geq 10$ points, the ``optimal coreset" seems to have all the $6$ isolated input points, including $m-6$ input points inside the large cluster, that are well distributed in this cluster. What would be the expected coresets of size $m$ using the existing techniques?
            \begin{figure}[htp]
                \centering
                \subfigure{\includegraphics[width=0.31\textwidth, frame]{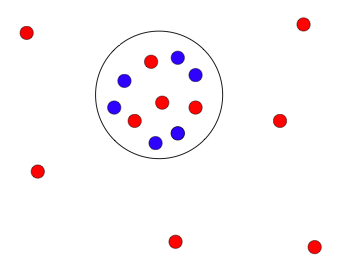}}
                \subfigure{\includegraphics[width=0.31\textwidth, frame]{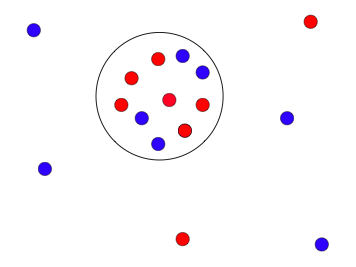}}
                \subfigure{\includegraphics[width=0.31\textwidth, frame]{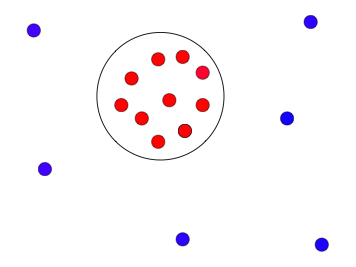}}
                \caption{\small \it A set of $16$ points that consists of $7$ clusters that are far from each other. Each of the first $6$ clusters contains a single point. The red points are the expected selected points for a coreset of $10$ points (with repetitions) using: \textbf{(left)} Algorithm 1, \textbf{(middle)} Non-uniform (importance/sensitivity) sampling, and \textbf{(right)} Uniform sampling.}
                \label{fig:Ex1}
            \end{figure}

\textbf{Algorithm~\algkmean.} would return exactly this ``optimal coreset", as the $m$-means of $P$ consists of the $6$ isolated clusters and the $m-6$ means of the large cluster. For the case $m=7$, the algorithm will pick exactly one representative in each cluster, as it is the $7$-means of $P$. See Fig.~\ref{fig:Ex1}(left).

\textbf{Uniform Sampling.} If the large cluster is sufficiently large, all the points in a uniform sample will be from this cluster, while all the other (singleton) $6$ clusters will be missed. Since these clusters are far away from each other, the approximation error will then be very large. Even for large sample size, uniform sample misses isolated clusters that are crucial for obtaining a small error. See Fig.~\ref{fig:Ex1}(right).

\textbf{Non-Uniform Sampling.} The optimal distribution that will make sure that a representative from each cluster will be selected to the coreset, is to sample a point from each of the $k=7$ clusters with roughly the same probability. However, due to the independent (i.i.d.) sampling approach, the number of samples that are needed in order to have a representative from each cluster is more than $k=7$. In general, the expected sample size is $O(k\log k)$. This phenomena is known as the coupon collector problem: if there is a coupon in each box at the supermarket, picked uniformly at random from a collection of $k$ distinct coupons, then one need to buy $O(k\log k)$ boxes in expectation to have the collection of all the $k$ coupons. This is compared to the deterministic construction of Algorithm~1 that always pick the desired $k$ representatives.

Even after having a representative from each of the isolated clusters a non-uniform sampling will keep sample a point from one of these clusters with probability $6/7$. This means that from the total ``budget" of $m$ points in the coreset, a large fraction will be used to sample the same point again and again. This is also why in Fig.~\ref{fig:Ex1}(middle) there is less number of red points than the other constructions.

\section{Practical and Simple Boosting of Existing Heuristics}\label{prac}
To get the desired phenomena that is described in Section~\ref{secphen} there is no need to actually compute the $m$-means for many values of $m$ and existing heuristics can be used. For example, any reasonable $k$-means heuristics for $k\geq 7$ would yield a set of $k$ red points as in Fig.~\ref{fig:Ex1}(left).

\textbf{The chicken and the egg phenomena.} As in Algorithm 1, coresets for solving optimization problems usually need to solve the optimization problem in order to decide which points are important to sample. This problem is solved in theory using rough approximation algorithms or bi-criteria approximations~\cite{edo, FL11} that replaces the optimal solution, or using the merge-and-reduce tree that apply the coreset constuction only on small sets.
In practice, algorithms that compute provable $(1+\eps)$-approximations or even $2$-approximations for the $k$-means clustering are rarely used. Instead, heuristics such as the classic LLoyd's $k$-means or $k$-Means++~\cite{ostrovsky2006effectiveness, arthur2007k} are used.

Based on our experimental results, a rough approximation using existing heuristics seems to be suffices.
In addition, plugging $\eps$ as an input, almost always produces coreset with error that is much smaller than $\eps$. This is common also in other coresets and related to the facts that the analysis is (i) for the worst case input set and not a specific $P$ that is usually well structured, (ii)  sub-optimal compared to the actual error, (iii) consider every set of $k$ centers, while we usually care about the optimal solution under some constraints.

We suggest to take the coreset size $|S|$ as the input and run $m=1,\cdots,|S|$ iterations of our algorithm. In fact, our experimental results suggest the following simple approach that use a single instead of $|S|$ runs and yields only slightly less better results.

\textbf{Boosting technique.} Given a heuristic for solving the $k$-means problem using some $N\gg k$ iterations, Algorithm 1 suggests to run the heuristic only small number of $m=O(k)\ll N$ iterations. Then, we take the mean of each of the $m$ clusters, weighted by the size of the cluster, or construct a $(1,\eps)$-coresets on each of the $m$ clusters. Then, we run the heuristic $N$ times on this ``coreset" of size $m$. Even if each iteration of the heuristics takes linear time of $O(nd)$, the running time is reduced from $O(Nnd)$ to $O(mnd)=O(knd)$.

\textbf{Example 1: LLyod'$k$-means.} In the case of Lloyd's $k$-means, each iteration takes $O(nd)$ time for computing the distances from the existing seed of $k$ centers, and then $O(nd)$ time is needed to compute the next set of centers. The boosting technique above suggests to run $m\sim O(k)\ll N$ such iterations on $P$ to produce a weighted coreset of size $m$. Then run the $N$ iterations on this coreset.

\textbf{Example 2: KMean++.} The KMean++ algorithm picks another point to the output set in each iteration, where the first point is a random seed. The next point is sampled with probability that is proportional to the distances of the input points to the center (points) that were already picked. This is very similar to the importance sampling that is used for constructing existing coresets with a crucial difference: the sampling is not independent and same for each new point, but \emph{adaptive}, i.e., based on points that were already picked. This is exactly the advantage of our approach compared to the non-uniform sampling, as described in Fig.~\ref{fig:Ex1} and Section~\ref{secphen}. Note that KMean++ will always select the right centers in Fig.~\ref{fig:Ex1}, no matter what is the seed and although it is a random algorithm.

The KMean++ algorithm is very natural for using with our boosting technique: We just run it for $m$ iterations to get a coreset of size $m$. Then we run KMean++ $N\gg m$ times on the coreset. In each of the $N$th times we use a different seed (first point) and take the optimal among the $N$ sets of $k$-mean candidates. Line~3 of Algoirthm~1 suggests an interesting way to choose the size $m$ of the coreset, based on our analysis in the supplementary material.

\section{Experimental Results}\label{sec:xp}
    \textbf{\bf Datasets.} To produce experimental results we have use two well known datasets.

    {\bf MNIST handwritten digits\cite{mnist}.} The MNIST dataset consist of $n=60,000$ grayscale images of handwritten digits. Each image of size 28x28 pixels was transformed to the the vector row of $d=784$ dimensions.

    {\bf Pendigits\cite{pendigits}.} This is a dataset from the UCI repository. The dataset created out of 250 samples provided by 44 writers. These writers were asked to write 250 digits in random order inside boxes of 500 by 500 tablet pixel resolution. The tablet sends $x$ and $y$ tablet coordinates and pressure level values of the pen at fixed time intervals (sampling rate) of 100 miliseconds. Digits are represented as constant length feature vectors of size $d=16$ the number of digits in the dataset is $n=10992$.

    {\bf NIPS dataset\cite{nipsOnline}.} The OCR data from the collection which represents 13 years of NIPS proceedings. The overall of 15,000 pages and 1958 articles. For each author there is a words count vector extracted, where ith entry in the vector represents count of the particular word which was used in one of the conference submissions by given author. There are overall $n=2865$ authors and words corpus size is $d=14036$.

    \textbf{Expirement.} We used our algorithm to boost the performance of Lloyd's $k$-means heuristic as explained in Section~\ref{prac}. Give a coreset size $m$ we run this heuristic for only $3$ iterations with $m$ centers.
    We compared our algorithm with uniform and importance sampling algorithms using both offline computation setting and streaming data model. For offline computation we used datasets above to produce coresets of sizes $ 100\leq m\leq 1500$, then computed $k$-means with values of $k=10, 15, 20, 25$ using Lloyd's heuristic for many iterations till convergence. While to simulate streaming data model we divided datasets into chunks and computed coresets of sizes $10 \leq m \leq 500$ using map-and-reduce techniques to construct a coreset tree, later repeated computation of $k$-means for same values of $k$.

    For each set of $k$ centers that was produced, we computed sum of squared distances to the original (full) set of points, and denoted these ``approximated solutions" by $C_{1}, C_{2}$ and $C_{3}$ for uniform, non uniform sampling and our algorithm respectively.
    The ``ground truth" or ``optimal solution" $C^k$ was computed using $k$-means on entire dataset until convergence. The empirical estimated error $\eps$ is then defined to be $\epsilon=C_{t}/C_{k} - 1$ for coreset number $t=1,2,3$.

    \textbf{Results for datasets}
        Fig.\ref{fig:offline} and Fig.\ref{fig:streaming} shows results for offline setting and streaming models respectevly. The results of out algorithm are outperforms the uniform sampling and non-uniform sampling algorithms. Important to note, that our algorithm starts with very small error value compared to others and improves error value gradually with sample size, while two others starts with greater error values and succeeds to converge to significantly smaller values only for large sample subsets.

    Fig.\ref{fig:boxplot_offline} and Fig.\ref{fig:boxplot_streaming}, shows the boxplot of error distribution for all three algorithms in offline and streaming settings. It's easy to see that that our algorithm show little variance across all experiments and mean error value is very close to the median, indicating that our algorithm produces very stable results, while running on streaming data whiskers are broader due to the multiplicative factor of $\log n$.

    In Fig.\ref{FigMem} we present the memory (RAM) footprint during the coreset construction based on synthetically generated random data.  These results are common to other coresets papers.  The oscillations corresponds to the number of coresets in the tree that each new chunk needs to update. For example, the first point in a streaming tree is updated in $O(1)$,  however the $2^i$th point for some $i \geq 1$ climbs up through $O(\log i)$ levels in the tree, so $O(\log i)$ coresets need to be merged.

            \begin{figure*}
                \subfigure{\includegraphics[width=0.3\textwidth]{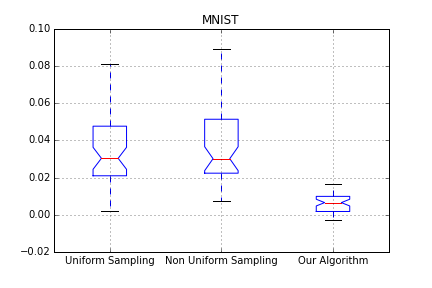}}
                \subfigure{\includegraphics[width=0.3\textwidth]{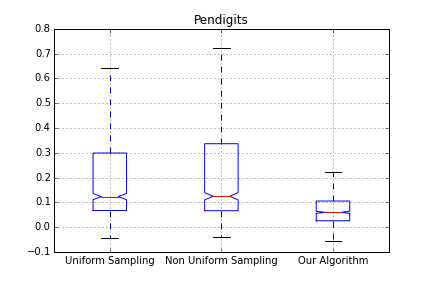}}
                \subfigure{\includegraphics[width=0.3\textwidth]{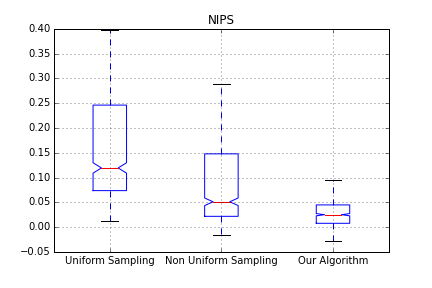}}
                \caption{\small \it Error (y-axis) box-plots for real-data sets, ofline computation model.}
                \label{fig:boxplot_offline}
            \end{figure*}

            \begin{figure*}
                \subfigure{\includegraphics[width=0.3\textwidth]{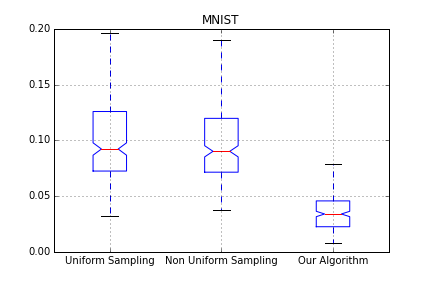}}
                \subfigure{\includegraphics[width=0.3\textwidth]{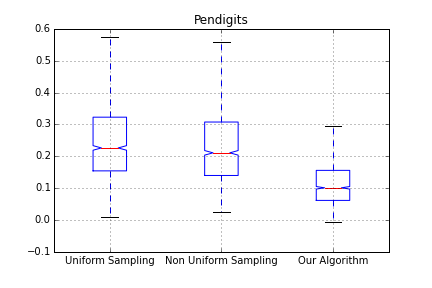}}
                \subfigure{\includegraphics[width=0.3\textwidth]{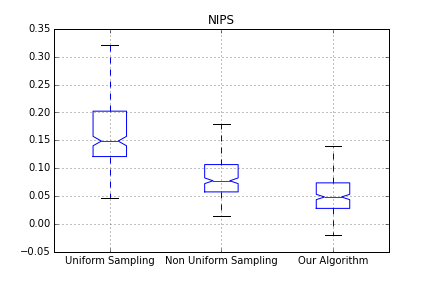}}
                \caption{\small \it Error (y-axis) box-plots for real-data sets, streaming computation model.}
               \label{fig:boxplot_streaming}
            \end{figure*}

            \begin{figure*}
                \subfigure[MNIST, k=10]{\includegraphics[width=0.3\textwidth]{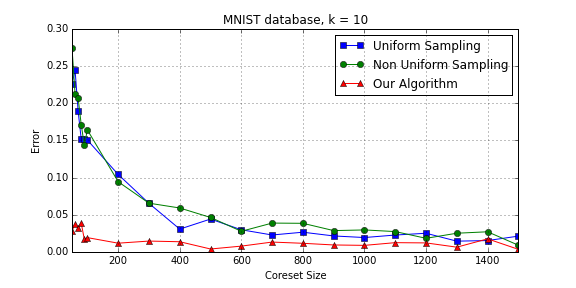}}
                \subfigure[Pendigits, k=10]{\includegraphics[width=0.3\textwidth]{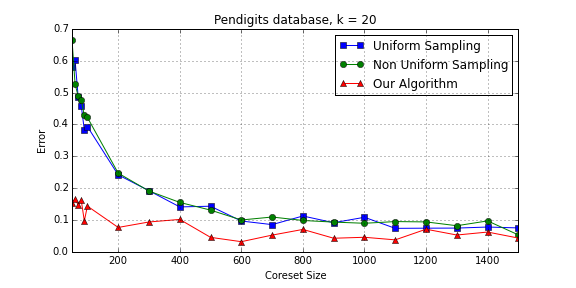}}
                \subfigure[NIPS, k=5]{\includegraphics[width=0.3\textwidth]{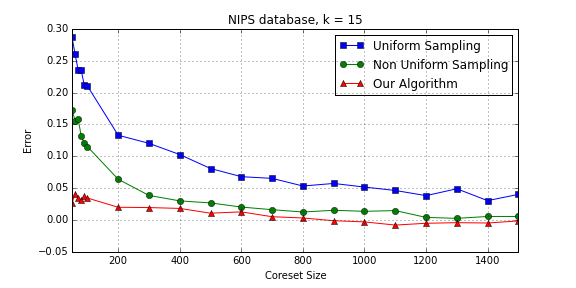}}

                \subfigure[MNIST, k=15]{\includegraphics[width=0.3\textwidth]{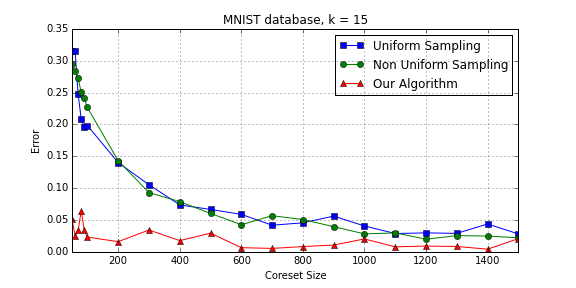}}
                \subfigure[Pendigits, k=15]{\includegraphics[width=0.3\textwidth]{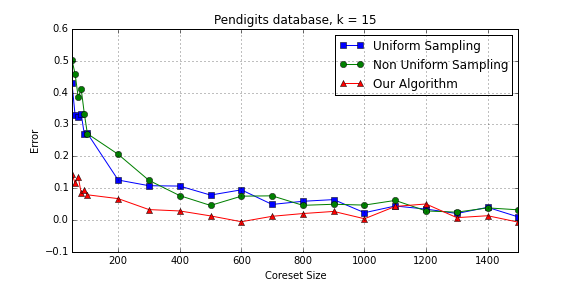}}
                \subfigure[NIPS, k=10]{\includegraphics[width=0.3\textwidth]{NIPS_10.png}}

                \subfigure[MNIST, k=20]{\includegraphics[width=0.3\textwidth]{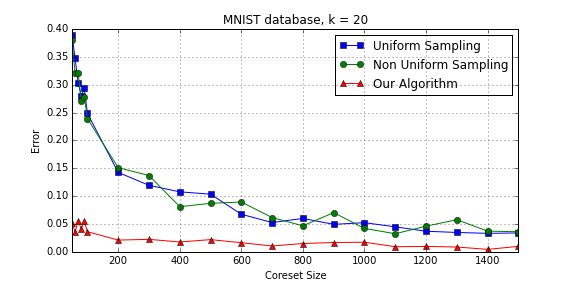}}
                \subfigure[Pendigits, k=20]{\includegraphics[width=0.3\textwidth]{Pendigits_20.png}}
                \subfigure[NIPS, k=15]{\includegraphics[width=0.3\textwidth]{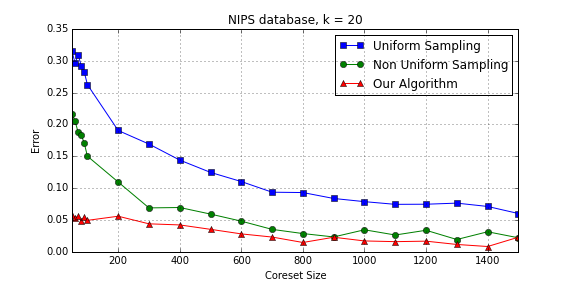}}

                \subfigure[MNIST, k=25]{\includegraphics[width=0.3\textwidth]{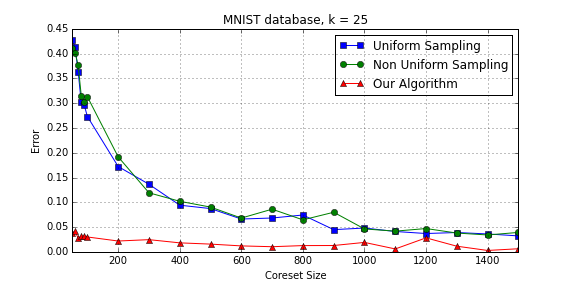}}
                \subfigure[Pendigits, k=25]{\includegraphics[width=0.3\textwidth]{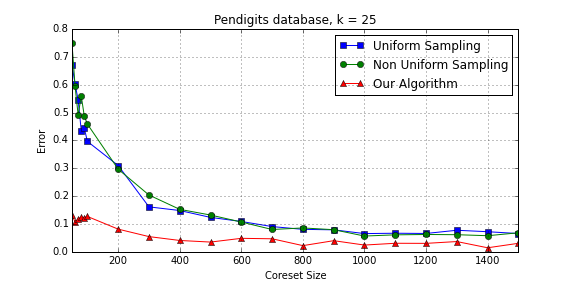}}
                \subfigure[NIPS, k=20]{\includegraphics[width=0.3\textwidth]{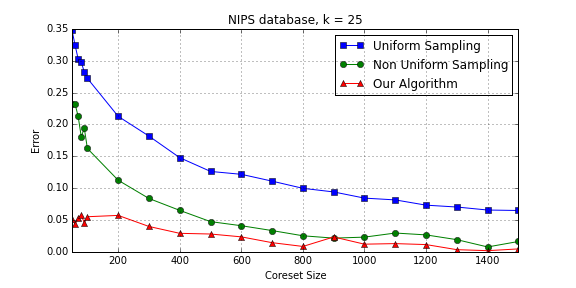}}

                \caption{\small \it Offline setup comparison of uniform sampling, non uniform sampling and our algorithms.}
                \label{fig:offline}
            \end{figure*}

            \begin{figure*}
                \subfigure[MNIST, k=10]{\includegraphics[width=0.3\textwidth]{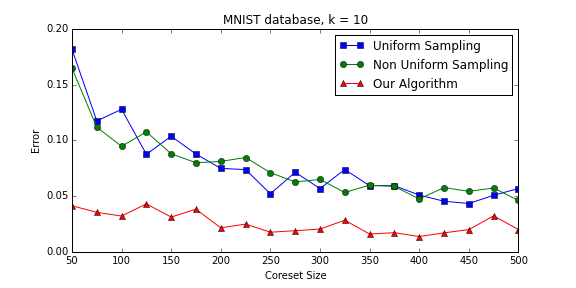}}
                \subfigure[Pendigits, k=10]{\includegraphics[width=0.3\textwidth]{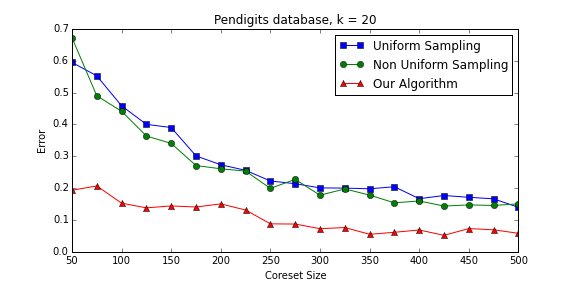}}
                \subfigure[NIPS, k=5]{\includegraphics[width=0.3\textwidth]{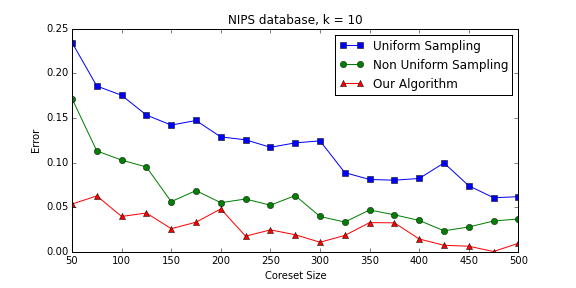}}

                \subfigure[MNIST, k=15]{\includegraphics[width=0.3\textwidth]{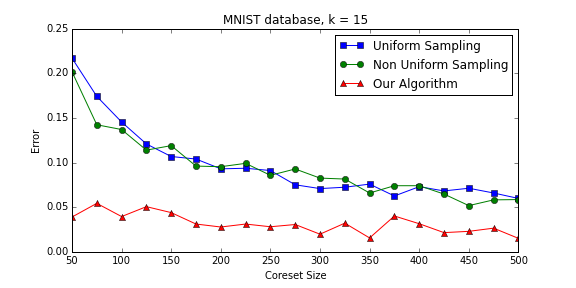}}
                \subfigure[Pendigits, k=15]{\includegraphics[width=0.3\textwidth]{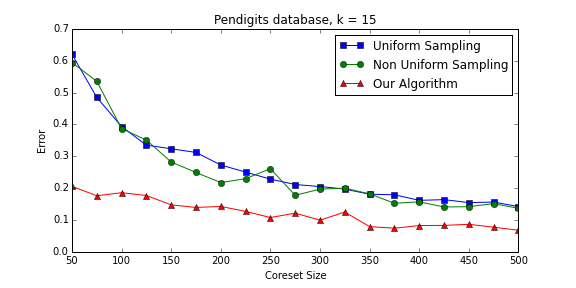}}
                \subfigure[NIPS, k=10]{\includegraphics[width=0.3\textwidth]{NIPS_10_streaming.png}}

                \subfigure[MNIST, k=20]{\includegraphics[width=0.3\textwidth]{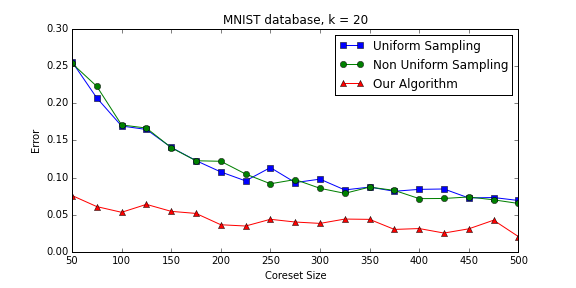}}
                \subfigure[Pendigits, k=20]{\includegraphics[width=0.3\textwidth]{Pendigits_20_streaming.png}}
                \subfigure[NIPS, k=15]{\includegraphics[width=0.3\textwidth]{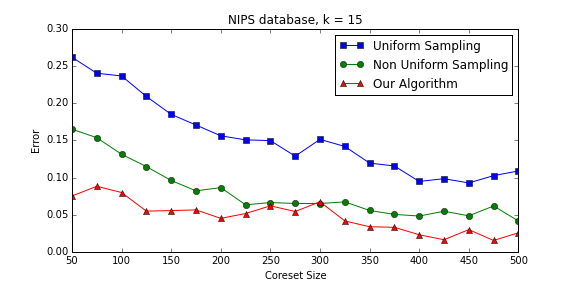}}

                \subfigure[MNIST, k=25]{\includegraphics[width=0.3\textwidth]{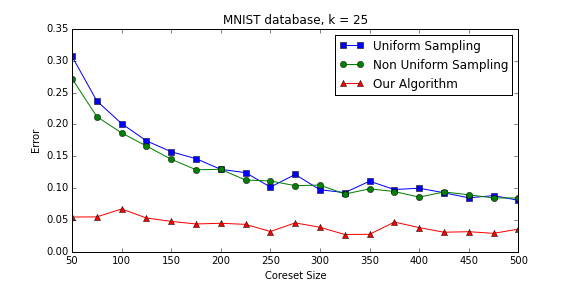}}
                \subfigure[Pendigits, k=25]{\includegraphics[width=0.3\textwidth]{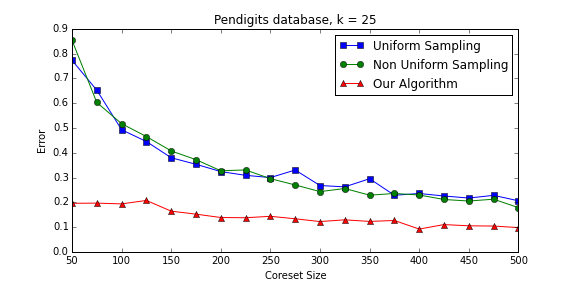}}
                \subfigure[NIPS, k=20]{\includegraphics[width=0.3\textwidth]{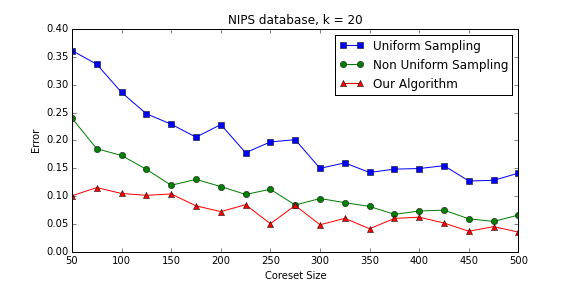}}
                \caption{\small \it Streaming setup comparison of uniform sampling, non uniform sampling and our algorithms.}
                \label{fig:streaming}
            \end{figure*}

			\begin{figure}
                \begin{center}
                \includegraphics[height=3cm, width=9cm]{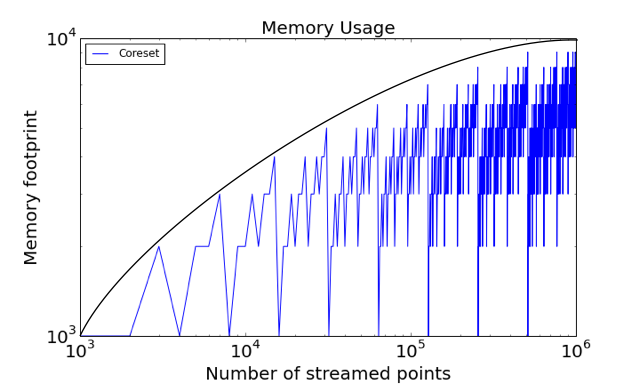}
			   \caption{\small \it  Allocated memory (y-axis) grows logarithmically during streaming coreset construction. The Zig-zag patterns caused by the binary merge-reduce tree in Fig.~\ref{fig:tree}.}
			   \label{FigMem}
                \end{center}
			\end{figure}

\newpage
\bibliographystyle{abbrv}
\bibliography{mybib}
\end{document}